\documentclass[twoside,leqno,twocolumn]{article}

\usepackage[letterpaper]{geometry}

\usepackage{ltexpprt}
\usepackage{hyperref}

\usepackage{caption}
\usepackage{subcaption}

\usepackage{algorithm}
\usepackage{algpseudocode}

\usepackage{multirow}
\usepackage{booktabs}

\usepackage{tikz}
\usetikzlibrary{positioning}

\usepackage{amsmath, amssymb, graphicx}
\DeclareMathOperator*{\argmax}{arg\,max}
\DeclareMathOperator*{\argmin}{arg\,min}

\begin{document}

\newtheorem{definition}{Definition}

\newcommand\relatedversion{}

\title{\Large Sensor Placement for Learning in Flow Networks \relatedversion}
\author{Arnav Burudgunte\thanks{Rice University, \{ab141,arlei\}@rice.edu}
\and Arlei Silva\footnotemark[1]}

\date{}

\maketitle

% Copyright Statement
% When submitting your final paper to a SIAM proceedings, it is requested that you include
% the appropriate copyright in the footer of the paper.  The copyright added should be
% consistent with the copyright selected on the copyright form submitted with the paper.
% Please note that "20XX" should be changed to the year of the meeting.

% Default Copyright Statement
\fancyfoot[R]{\scriptsize{Copyright \textcopyright\ 2024 by SIAM\\
Unauthorized reproduction of this article is prohibited}}

% Depending on which copyright you agree to when you sign the copyright form, the copyright
% can be changed to one of the following after commenting out the default copyright statement
% above.

%\fancyfoot[R]{\scriptsize{Copyright \textcopyright\ 20XX\\
%Copyright for this paper is retained by authors}}

%\fancyfoot[R]{\scriptsize{Copyright \textcopyright\ 20XX\\
%Copyright retained by principal author's organization}}

%\pagenumbering{arabic}
%\setcounter{page}{1}%Leave this line commented out.

\begin{abstract} \small\baselineskip=9pt Large infrastructure networks (e.g. for transportation and power distribution) require constant monitoring for failures, congestion, and other adversarial events. However, assigning a sensor to every link in the network is often infeasible due to placement and maintenance costs. Instead, sensors can be placed only on a few key links, and machine learning algorithms can be leveraged for the inference of missing measurements (e.g. traffic counts, power flows) across the network. This paper investigates the sensor placement problem for networks. We first formalize the problem under a flow conservation assumption and show that it is NP-hard to place a fixed set of sensors optimally. Next, we propose an efficient and adaptive greedy heuristic for sensor placement that scales to large networks. Our experiments, using datasets from real-world application domains, show that the proposed approach enables more accurate inference than existing alternatives from the literature. We demonstrate that considering even imperfect or incomplete ground-truth estimates can vastly improve the prediction error, especially when a small number of sensors is available. 
\end{abstract}\\


\section{Introduction}

This paper addresses the problem of effectively placing sensors to estimate edge flows in networks. Given a flow value at each edge, our goal is to choose a set of edges to monitor in order to infer values at unmonitored edges. Versions of our problem appear in infrastructure networks, such as traffic \cite{activelearning}, water \cite{outbreakdetection}, and power \cite{samudrala2019power} networks. We assume that each edge in the network has some associated value (e.g. traffic flow, water flow, electrical current) that must be measured to track congestion, anomalies, and other events. Owing to the size of the network and measurement cost, we often cannot place a sensor at every edge. A common solution is to place sensors at a small subset of the edges and estimate the rest using semi-supervised learning \cite{sampling, activelearning, spectraltransforms}. Because the choice of this subset can seriously alter the prediction quality, our goal is to efficiently select a subset that yields the most accurate estimate. We focus here on values obeying the assumption \textit{flow conservation}, meaning the sum of flows into a vertex is (approximately) equal to the sum of flows out of it. 

The sensor placement problem is related to two machine learning paradigms. The first is \textit{semi-supervised learning (SSL)}, in which models learn from both labeled and unlabeled data. SSL is motivated by applications where the availability of labeled data is limited \cite{graphbasedssl, labelprop}. Graph-based SSL is useful when a dataset can be encoded as a graph with labels on nodes or edges; the unlabeled data can be leveraged by propagating labels through the graph structure.  

The second machine learning paradigm related to sensor placement is \textit{active learning}, in which a model iteratively makes requests for specific observations to be labeled \cite{settles2009active,cohn1996active,cohn1994improving}. When applied to graph data, active learning can take advantage of the graph structure to select samples to be labeled more effectively \cite{bilgic2010active,cesa2010active,gu2012errormin}. 
Such algorithms often assume no prior information about the ground-truth labels at first and slowly learn more about the network as more data is labeled. In practice, however, we often have outside information about the edge labels before placing any sensors. For example, urban planners have detailed predictions for traffic flow based on traffic demand models, spatial and temporal data, and other factors \cite{trafficdemand, SUMO, li2021spatial}. Such estimates can be extremely useful for inferring labels with a small number of sensors because they significantly constrain the possible solutions---there might be many possible predictions that satisfy the flow conservation assumption, but only a few will correctly predict the volume of traffic at major roads. As we will show, even an imperfect estimate of edge values can generate a more effective choice of sensors than a naive algorithm that makes predictions based on flow conservation alone.

%The question we address in this paper is \textit{how to use existing knowledge of edge values (or flows) to improve our predictions}.
The question we address in this paper is \textit{how to select a small number of edge flows to be monitored such that the error of flow predictions for the remaining of the network is minimized}. As the problem is NP-hard, we propose an efficient algorithm that greedily selects sensors to minimize the error of a flow conservation-based prediction, combining knowledge of the network topology and ground-truth flow values. However, the efficiency of our algorithm depends on the efficient evaluation of possible sensor locations. We achieve this objective via two optimizations, the lazy evaluation of candidate sensors and the fast evaluation of the benefit of adding each sensor using the Woodbury matrix identity. Our experiments evaluate the proposed algorithm in terms of accuracy and running time and show that it outperforms several alternatives. %By strategically choosing sensor locations, we are able to reconstruct nearly all of the flow in a network with a small fraction of edges.  

We summarize our main contributions as follows:
\begin{itemize}
\item We provide a formal statement of the sensor placement problem under flow conservation, with proof that the problem is NP-complete;
\item We propose an efficient greedy heuristic for sensor placement that combines two optimizations to speed up the computation of each greedy decision;
\item We compare the proposed heuristic against several alternatives using real networks and the experimental results show that our approach is more accurate while the proposed optimizations lead to significant savings in computing time.

%We present experimental results showing our approach's improvement over existing baselines.
\end{itemize}

\section{Related Work}

\paragraph{\textbf{Graph-Based Semi-Supervised Learning}} Graph-based SSL is a special case of SSL where observations can be encoded as vertices and edges represent relations between them \cite{spectraltransforms}. Label propagation \cite{labelprop, labelprop2} is the most popular technique for reconstructing smooth values. More recently, flow-based SSL \cite{activelearning} was proposed to infer conserved flows on the edges rather than smooth vertex labels. However, the conservation-based solution can be highly dependent on the choice of labeled edges and the nature of the data \cite{labelprop, activelearning,da2020combining}, which is a motivation for our work. We focus on the problem of selecting edges to be labeled (or sensor locations, in the case of infrastructure networks), which is a special case of the active learning problem \cite{settles2009active}. 

\paragraph{\textbf{Active Learning on Graphs}} With the exception  of \cite{activelearning}, active learning on graphs has been studied exclusively on the vertices under a smoothness assumption. Cut-based methods \cite{guillory2009label}, error bound minimization \cite{gu2012errormin}, and sampling (signal processing) \cite{samplingtheory} methods have been proposed as active learning strategies. In all these algorithms, the choice of nodes or edges is made using only the network topology, and thus cannot be optimized using the ground-truth values \cite{guillory2009label, activelearning}. %Such methods have been successful because the space of vertex signals has an intuitive basis----the eigenvectors of the Laplacian----which can be sampled to enforce the smoothness assumption. 
Such methods apply the eigenvectors of the Laplacian matrix as a basis for vertex signals to sample vertices as to enforce smoothness. 
Conversely, the corresponding basis for edge flows----the singular vectors of the incidence matrix----has many vectors with frequency 0, making it difficult to prioritize samples. We thus assume that edges are chosen by optimizing a reconstruction loss. We show that this approach allows us to match the performance of existing baselines using far fewer sensors. 

\paragraph{\textbf{Traffic Forecasting}} The problem of forecasting traffic flows, speeds, and other attributes is can be motivated by smart city applications. Early attempts at traffic prediction infer network flows from a set of origin-destination pairs \cite{trafficdemand} or vice versa \cite{Hazelton}; the popular traffic simulator SUMO \cite{SUMO} operates using a similar principle. More recently, deep learning has been used to predict traffic as a function of spatial and temporal data  \cite{li2021spatial, peng2020spatial, song2020spatial}. The relevant result for our work is that there are sophisticated, accurate models available for predicting traffic flow in graphs even without labels from sensors. We show that even when such estimates are noisy or imperfect, they provide useful information about ground-truth flows. We focus on the spatial dimension of the problem and use simpler inference models such as label propagation, but our algorithms could be generalized to use any forecasting model. 

\paragraph{\textbf{Sensor Placement and Influence Maximization}} Previous work has studied sensor placement for other objectives such as monitoring spatial phenomena \cite{krause2006near} and detecting contaminants in water distribution networks \cite{krause2008efficient, outbreakdetection}. Similar to these problems, we study sensor placement as a combinatorial optimization problem, though we minimize a prediction error rather than a fixed penalty function. Closely related to sensor placement is the influence maximization problem, which asks for the best subset of nodes to target for influence such that after some diffusion process, the maximum number of nodes have been influenced \cite{kempe2003maximizing}. Solutions often depend on the diffusion model, which predicts the number of nodes affected by the chosen seed set \cite{li2018influence}. Instead of using a process, we apply a machine learning model to make predictions based on a validation set (i.e. set of edge observations). Moreover, we notice that, although we apply similar greedy algorithms as those previously used for sensor placement and influence maximization, our objective functions are not submodular \cite{nemhauser1978analysis}.

\section{Problem Definition}

\subsection{Preliminaries}

We represent a  network as an unweighted graph $G = (V, E)$ with a set $V$ of $n$ vertices and set $E$ of $m$ edges. The graph is represented by the adjacency matrix $\mathbf{A} \in \mathbb{R}^{n \times n}$, where $\mathbf{A}_{ij} = 1$ if an edge exists between vertices $i$ and $j$ and $\mathbf{A}_{ij} = 0$ otherwise. 

We are interested in the value given by vector $\mathbf{f} \in \mathbb{R}^m$ where $\mathbf{f}_i$ is the flow at edge $i$. We represent the values of a subset $E' \subset E$ as a vector $\mathbf{x}_{E'} \in \mathbb{R}^{|E'|}$. Thus a set of labeled (sensor) edges $S$ is represented by the corresponding value $\mathbf{f}_S \in \mathbb{R}^{|S|}$, and a set of unlabeled (target) edges $T$ is represented as $\mathbf{f}_T \in \mathbb{R}^{|T|}$. 

\subsection{General Sensor Placement Problem}

The sensor placement problem can be divided into two parts. The first part predicts the values of unlabeled edges based on a given set of sensors. The second is choosing a set of sensors that yields the prediction with the lowest error. We formalize both parts here. 

\subsubsection{Prediction}
\label{sec:prediction}

Given a labeled set of $k$ sensors $S$ and corresponding observed flows $\mathbf{f}_{S} \in \mathbb{R}^k$, we produce an estimate $\mathbf{\hat{f}}$ for $\mathbf{f}$ via graph-based semi-supervised learning. The divergence on each vertex $i$ is defined as the difference between flows out of $i$ and flows into $i$:

\begin{equation}
    \label{eqn::divergence}
    \text{div}(i) = \sum_{e \in E : e \text{ out of } i} \mathbf{f}_e - \sum_{e \in E : e \text{ into } i} \mathbf{f}_e
\end{equation}

Flows for missing edges are estimated by minimizing the sum of squared divergences given by

\begin{equation}
    \label{conservation}
    ||\mathbf{Bf}||^2 = \sum_{i \in V} (\text{div}(i))^2
\end{equation}
where the incidence matrix $\mathbf{B} \in \mathbb{R}^{n \times m}$ is defined as
\begin{equation}
\mathbf{B}_{ij} =
    \begin{cases}
        1 & \text{if edge } e_j \text{ enters node } i\\
        -1 & \text{if edge } e_j \text{ leaves node } i\\
        0 & \text{otherwise.}
    \end{cases}
\end{equation}
For undirected graphs, we first choose an arbitrary orientation for each edge before constructing $\mathbf{B}$. We then minimize the sum of the divergence and a regularization term parameterized by $\lambda \in \mathbb{R}_+$:

\begin{equation}
\begin{aligned}
    \mathbf{\hat{f}}^* &= \argmin_{\mathbf{\hat{f}} \in \mathbb{R}^{m}} ||\mathbf{B\hat{f}}||^2 + \lambda^2 \cdot ||\mathbf{\hat{f}}||^2 \\
    &\textrm{s.t. } \mathbf{\hat{f}}_{S} = \mathbf{f}_{S} 
    \label{flowprediction}
\end{aligned}
\end{equation}
where $\lambda$ guarantees that the solution is unique. % For directed graphs, we additionally impose the constraint that $\mathbf{\hat{f}}_i \geq 0 \text{ for all } i$. 

The resulting optimization problem can be rewritten as a regularized least squares. Define $\mathbf{f}_S^0 \in \mathbb{R}^m$ such that $\mathbf{f}^0_i = \mathbf{f}_i$ if $i \in S$ and $\mathbf{f}^0_i = 0$ otherwise. Let $\mathbf{H}_T \in \mathbb{R}^{m\times (m-k)}$ be a matrix (map) such that $\mathbf{H}_{ij}=1$ if flow $\mathbf{f}_i$ maps to $(\mathbf{f}_T)_j$ (i.e., they correspond to the same edge). The least-squares formulation is \cite{chen2010nonnegativity,da2020combining}:

\begin{equation}
    \label{eqn:flowleastsquares}
    \mathbf{f}_T^* = \argmin_{\mathbf{f}_T \in \mathbb{R}^{m-k}} ||\mathbf{B}\mathbf{H}_T\mathbf{f}_T-\mathbf{B}\mathbf{f}_S^0||^2+\lambda^2 \cdot  ||\mathbf{f}_T||^2 
\end{equation}

\subsubsection{Sensor Placement}

Given a set $T$ of target vertices, a set $C$ of candidate vertices, and budget $k$, our problem is to choose the subset of $k$ vertices in $C$ that yields the best prediction for $T$:

\begin{equation}
\begin{aligned}
    S^* = & \argmin_{S \subseteq C,  |S| = k} \quad ||\mathbf{\hat{f}}_T - \mathbf{f}_T||^2 \\
    \textrm{s.t.} \quad &\hat{\mathbf{f}}_T = \phi(\mathbf{f}_S, S, \lambda)
    \label{eqn::sensorplacement}
\end{aligned}
\end{equation}
where $\phi$ is the prediction model (see Section \ref{sec:prediction}).  

\subsection{Hardness}
\label{sec::hardness}

\begin{figure}
    \begin{center}
            \includegraphics[width=0.49\textwidth]{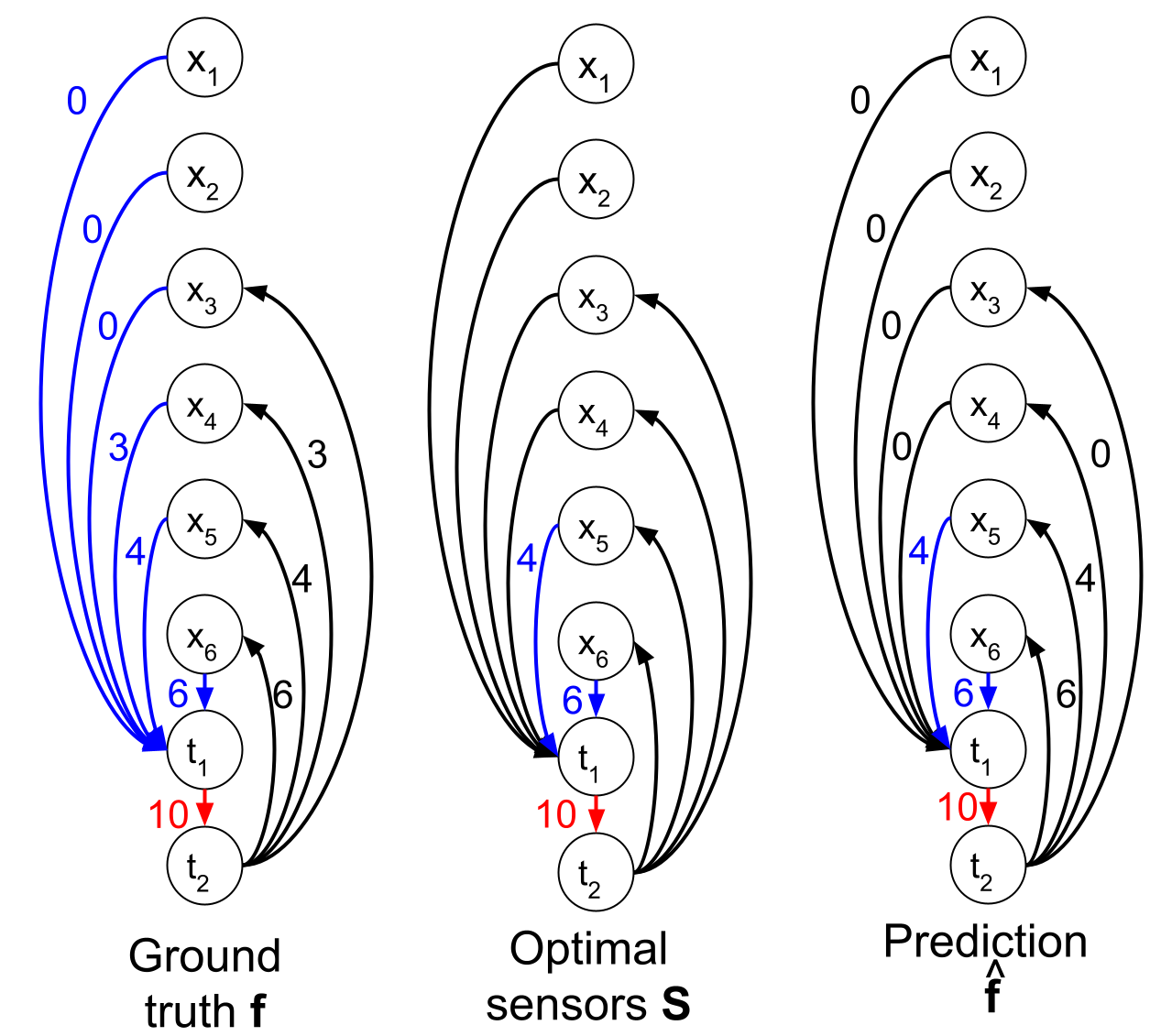}
    \end{center}
    \caption{Reducing \emph{SUM} to \emph{SENSOR}, with $X\!=\!\{3,4,6\}$ and $t\!=\!10$, candidate set $C\!=\!X$ (blue), and the target $T\!=\!\{(t_1, t_2)\}$ (red). Choosing the two edges that sum to $t$ for $S$ generates a perfect prediction for $T$.}
    \label{fig::reduction}
\end{figure}

We now formally define the sensor placement problem and show that it is NP-complete. 

\begin{definition}[The Sensor Placement Problem]
Given a graph $G$, a candidate set of edges $C \subseteq E$, a target set of edges $T \subseteq E$, a budget $k$, and an error $\epsilon$, \emph{SENSOR$(G, C, T, k)$} consists of determining whether there exists a set of edge labels $S \subseteq C$ such that $|S| = k$ and the edge predictions $\mathbf{\hat{f}}$ for $S$ (see Equation \ref{eqn::sensorplacement}) have an error $||\mathbf{\hat{f}}_T - \mathbf{f}_T || \leq \epsilon$.
\end{definition}

\begin{theorem}
Sensor Placement is NP-complete. 
\label{thm:hardness}
\end{theorem}
\begin{proof}
Given a certificate $S \subseteq C$ of edges selected as sensors, we can clearly compute $\mathbf{\hat{f}}$ and check the error in polynomial time. Let \emph{SUM}$(X, t)$ be an instance of the subset sum problem for a finite set $X \subseteq \mathbb{Z}_+$ and $t \in \mathbb{Z}_+$. The problem is to find a subset $X' \subset X$ such that $\sum_{x \in X'}x = t$. This problem is NP-hard \cite{clrs}. 

We can reduce an arbitrary instance of \emph{SUM} to \emph{SENSOR} by constructing graph $G$ as follows. Create a vertex for each element in $X$ and two additional vertices $t_1$ and $t_2$. Add a set of edges $C = \{(x, t_1) : x \in X\}$, each with a flow equal to $s$; an edge $(t_1, t_2)$ with flow $t$; and set of edges $\{(t_2, x): x \in X\}$ each with flow $x$. Finally, add $|X|$ vertices with 0 flow into $t_1$. 

There exists a solution to \emph{SENSOR}$(G, C, \{(t_1, t_2)\}, k)$ with prediction error $\epsilon = 0$ iff one exists for $\emph{SUM}(X, t)$. The equivalence is straightforward: if some set $X'$ sums to $t$, choose edges $\{(x', t_1): x' \in X'\}$, whose flows also sum to $t$, as sensors. If $|X'| < k$, choose edges with 0 flow for the remaining sensors. Since there is only edge out of $t_1$, that edge must have flow $t$. Conversely, if $k$ sensors correctly predict the flow on $(t_1, t_2)$, their flows must sum to $t$, and thus provide a solution to \emph{SUM}. 
\end{proof}

\subsection{Submodularity}
\label{sec::submodularity}

The submodularity property \cite{nemhauser1978analysis} is often applied to provide approximation guarantees for poly-time algorithms that solve sensor placement and other NP-hard related problems on graphs \cite{kempe2003maximizing,outbreakdetection}. However, it is straightforward to show that our sensor placement objective is not submodular. This is because the objective from Equation \ref{eqn::sensorplacement} is not monotonically decreasing. The addition of a new sensor might increase the error of the flow predictions. For example, consider a triangle where only one edge $e$ has a nonzero flow. With no sensors, Equation \ref{eqn:flowleastsquares} predicts $\mathbf{\hat{f}} = 0$, resulting in an error of $|\mathbf{f}_e|^2$. After choosing $e$ as a sensor, the prediction becomes $\mathbf{f}_e$ everywhere, and the error increases to $2|\mathbf{f}_e|^2$. One could consider a simplified version of our problem as a maximum flow coverage problem, where the goal is to place the sensors as to maximize the total flow covered by the sensors. This simplified version of the problem is submodular.

\section{Methods}

\subsection{Greedy Heuristic}

Given our hardness result (Theorem \ref{thm:hardness}), we propose a greedy algorithm that iteratively selects the sensor that minimizes the prediction error in Equation \ref{eqn::sensorplacement}. The pseudocode is given in Algorithm \ref{alg::greedy}. At each of the $k$ iterations (lines 2-11), the algorithm selects the best new sensor $s^*$ based on its resulting prediction error $\epsilon$ (lines 6-10).

\begin{algorithm}[t!]
\caption{Sensor Placement($G, C, T, k, \mathbf{f}$)}
\label{alg::greedy}
\begin{algorithmic}[1]
\State $S \gets \emptyset$

\For{$i = 1,...,k$}
    \State $\epsilon_{min} \gets \infty$
    \ForAll{$s \in C-S$}
        \State $S'=S\cup \{s\}$
        \State $\epsilon = ||\phi(\mathbf{f}_{S'}, S', \lambda) - \mathbf{f}_T||_2^2$
        \If{$\epsilon <  \epsilon_{min}$}
            \State $\epsilon_{min} \gets \epsilon$
            \State $s^* = s$
        \EndIf
    \EndFor
    \State $S \gets S \cup \{s^*\}$
\EndFor

\noindent \Return $S$
\end{algorithmic}
\end{algorithm}

%Something about approximation.
The flow prediction problem is not submodular (see Section \ref{sec::submodularity}). In practice, however, it is almost always possible to add a sensor that decreases the prediction error, and thus the problem exhibits (empirical) diminishing return (see Figure \ref{fig::flowcorrelations}). This property enables our algorithm to work well in practice.

%Complexity
At each of the $k$ iterations, our algorithm evaluates the benefit of every not yet chosen edge in $C$. This requires $O(m)$ solutions of the least-squares problem in Equation \ref{eqn:flowleastsquares}. Each instance takes $O(m^3)$ time to solve, giving an overall running time of $O(m^5)$. We now propose two strategies to speed up this algorithm. 

%Maybe we can state in the hardness section that the problem is not submodular as a claim. Then we can expand this section with one extra paragraph giving intuition on why the problem is "almost submodular".
\subsection{Lazy Evaluation}
\label{sec:lazyevaluation}
Previous works have proposed lazy evaluation to speed up greedy algorithms for submodular optimization \cite{outbreakdetection}. As discussed in Section \ref{sec::submodularity}, our sensor placement problem is not submodular but the associated coverage problem---maximizing the total flow coverage---is submodular. Based on this ``near-submodularity'' property, we propose applying lazy evaluation to speed up Algorithm \ref{alg::greedy}. More specifically, we will assume that (1) adding a sensor does not increase the flow prediction error and (2) the benefit of adding a sensor at a later iteration does not increase its benefit (i.e. its flow prediction error reduction). 

%First, we reduce the number of required estimates by taking advantage of the problem's ``near-submodularity''---as more sensors are placed, the benefit from not chosen edges rarely increases. This allows us to evaluate these benefits lazily \cite{outbreakdetection}. 

The lazy evaluation works as follows. Before choosing any sensors, we compute the benefit of each edge and store the results in a heap. At each iteration, we recompute the benefit of the best remaining edge. If it remains the best, then it is highly unlikely any other edge can overtake it, and we can choose that edge without reevaluating any other benefits. This saves up to $|C|$ evaluations per iteration. In case the evaluation decreases the rank of the top edge, we evaluate the next edge returned by the heap until the rank is unchanged.

%This section should be improved. The optimization is useful not only for updating but also for evaluating a new candidate edge. We should give more intution on why this works.
\subsection{Recursive Computation}
\label{sec:recursivecomputation}
We also propose a recursive computation to speed up the predictions themselves after a new labeled edge is revealed. The goal is to avoid solving Equation \ref{eqn:regularizedflowprediction} from scratch after a single new sensor is introduced. Instead, we will re-use results from previous iterations to quickly update flow predictions as new sensors are selected.

Our approach is similar to the recurrence for vertex labels given in \cite{zhu2005ssl} but works edge flows. To simplify the notation, let $\mathbf{X}_T = \mathbf{B}\mathbf{H}_T$. Then we can rewrite the closed form solution to Equation \ref{eqn:flowleastsquares} as 
\begin{equation*}
    \label{eqn:regularizedflowprediction}
    \mathbf{\hat{f}}_T = -(\mathbf{X}_T^T\mathbf{X}_T + \lambda^2\mathbf{I})^{-1}\mathbf{X}_T^T\mathbf{B}\mathbf{f}_S^0
\end{equation*}
which requires a matrix inversion.

We first compute an LU-decomposition of $\mathbf{X}_T^T\mathbf{X}_T + \lambda^2\mathbf{I}$ and use it to solve for $\mathbf{\hat{f}}_T$. Now suppose we select edge $e_i = e_{T_j}$, updating $S' \gets S \cup \{e_i\}$ and $T' \gets T - \{e_i\}$. The update is equivalent to removing the $j$th column of $\mathbf{X}_T$ or removing the $j$th row and column of $\mathbf{X}_T^T\mathbf{X}_T$. This is a rank two update and can be written in terms of matrices $\mathbf{U}, \mathbf{V} \in \mathbb{R}^{t \times 2}$ ($t = |T|$) defined as
\begin{equation*}
    \mathbf{U} = \begin{pmatrix}
        -1 & 0 \\
        0 & (\mathbf{X}_T^T\mathbf{X}_T)_{j, 1} \\
        \vdots & \vdots \\
        0 & (\mathbf{X}_T^T\mathbf{X}_T)_{j, t} \\
    \end{pmatrix}, \mathbf{V} = \begin{pmatrix}
        0 & -1\\
        (\mathbf{X}_T^T\mathbf{X}_T)_{1, j} & 0\\
        \vdots & \vdots \\
        (\mathbf{X}_T^T\mathbf{X}_T)_{t, j} & 0\\
    \end{pmatrix}
\end{equation*}
Further, let us define the downsampling matrix $\mathbf{S}_{j} \in \mathbb{R}^{(t-1) \times t}$ such that $\mathbf{S}_j\mathbf{X}_T$ removes the $j$th row of $\mathbf{X}_T$. It follows that
\begin{equation*}
    (\mathbf{X}_{T'}^T\mathbf{X}_{T'})^{-1} = \mathbf{S}_{j}(\mathbf{X}_{T}^T\mathbf{X}_{T} + \mathbf{U}\mathbf{V}^T)^{-1}\mathbf{S}_{j}^T
\end{equation*}

We now use the previously-computed LU-decomposition to solve for $\mathbf{Y} = (\mathbf{X}_T^T\mathbf{X} + \lambda^2\mathbf{I})^{-1}\mathbf{U}$ and $\mathbf{z} = (\mathbf{X}_{T}^T\mathbf{X} + \lambda^2\mathbf{I})^{-1}\mathbf{S}_j^T\mathbf{X}_{T'}^T\mathbf{B}\mathbf{f}_{S'}^0$. Finally, we apply the Woodbury identity \cite{woodbury} as follows:
\begin{equation*}
\begin{aligned}
    \mathbf{\hat{f}}_{T'} &= -(\mathbf{X}_{T'}^T\mathbf{X}_{T'} + \lambda^2\mathbf{I})^{-1}\mathbf{X}_{T'}^T\mathbf{B}\mathbf{f}_{S'}^0 \\
    &= -\mathbf{S}_j(\mathbf{X}_{T}^T\mathbf{X}_{T}  + \lambda^2\mathbf{I} + \mathbf{U}\mathbf{V}^T)^{-1}\mathbf{S}_j^T\mathbf{X}_{T'}^T\mathbf{B}\mathbf{f}_{S'}^0 \\
    &= -\mathbf{S}_j(\mathbf{z} - \mathbf{Y}(\mathbf{I} + \mathbf{V}^T\mathbf{Y})^{-1}\mathbf{V}^T\mathbf{z})
\end{aligned}
\end{equation*}

As a result, the only matrix that must be inverted directly is $\mathbf{I} + \mathbf{V}^T\mathbf{Y}$, but since this is a $2 \times 2$ matrix the inversion can be done in constant time. The LU-decomposition of $\mathbf{X}_T^T\mathbf{X}_T$, which is performed $O(m)$ times, has time complexity $O(m^3)$ \cite{bunch1974ludecomposition}. Solving a linear system using the decomposition takes $O(m^2)$ time, giving a new overall complexity of $O(m^4)$.  

%Maybe settings, datasets and baselines can be discussed here (if we have space)
\section{Experimental Results}
\label{sec:results}

We test our proposed algorithm on real-world traffic networks. %Experimental settings and metrics are discussed in the appendix. 
We consider two settings. First, we assume that our method is able to fully observe flow values to place the sensors, which is the ideal scenario. Next, we consider a more realistic setting where flows are unknown or only partially known.

\subsection{Experimental Settings}

\subsubsection{Datasets} We use traffic flows on road networks from Anaheim, Barcelona, Chicago, and Winnipeg (see Table \ref{tab:networks}) \cite{TransportationNetworks}. The nodes in each network represent intersections and edges represent roads between them. Each network is represented as a directed graph where the direction of an edge corresponds to the direction of traffic flow. The flows are generated 

\begin{table}
\begin{center}
  \caption{Road Networks Used for Experiments}
  \label{tab:networks}
  \begin{tabular}{ccc}
    \toprule
    Network & $n$ & $m$\\
    \midrule
    Anaheim & 416 & 914\\
    Barcelona & 1020 & 2522\\
    Chicago & 933 & 2950\\
    Winnipeg & 1052 & 2836\\
  \bottomrule
\end{tabular}
\end{center}
\end{table}

\subsubsection{Hyperparameters} We set $\lambda = 10^{-6}$ when solving the least-squares formulation for flow prediction. 

\begin{figure}[t!]
    \begin{center}
        \includegraphics[width=0.48\textwidth]{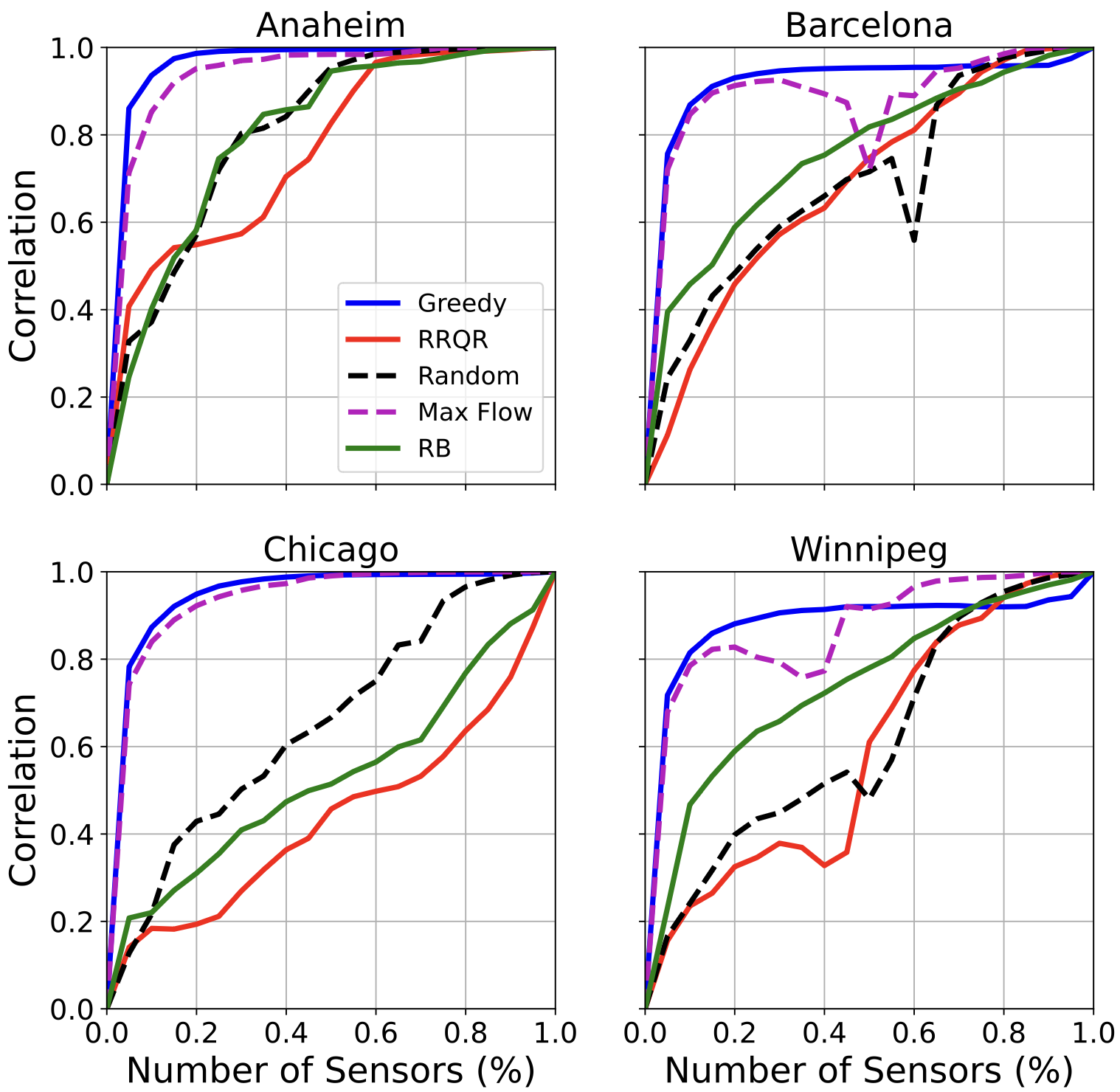}
    \end{center}
    \caption{Conserved flow prediction results when flows are fully observed for validation purposes. The plots show the correlation between the prediction $\mathbf{\hat{f}}$ and ground truth flow $\mathbf{f}$. Our greedy heuristic (Greedy) outperforms all four baselines in all datasets.}
    \label{fig::flowcorrelations}
\end{figure}

\begin{figure}[t!]
    \begin{center}
        \includegraphics[width=0.48\textwidth]{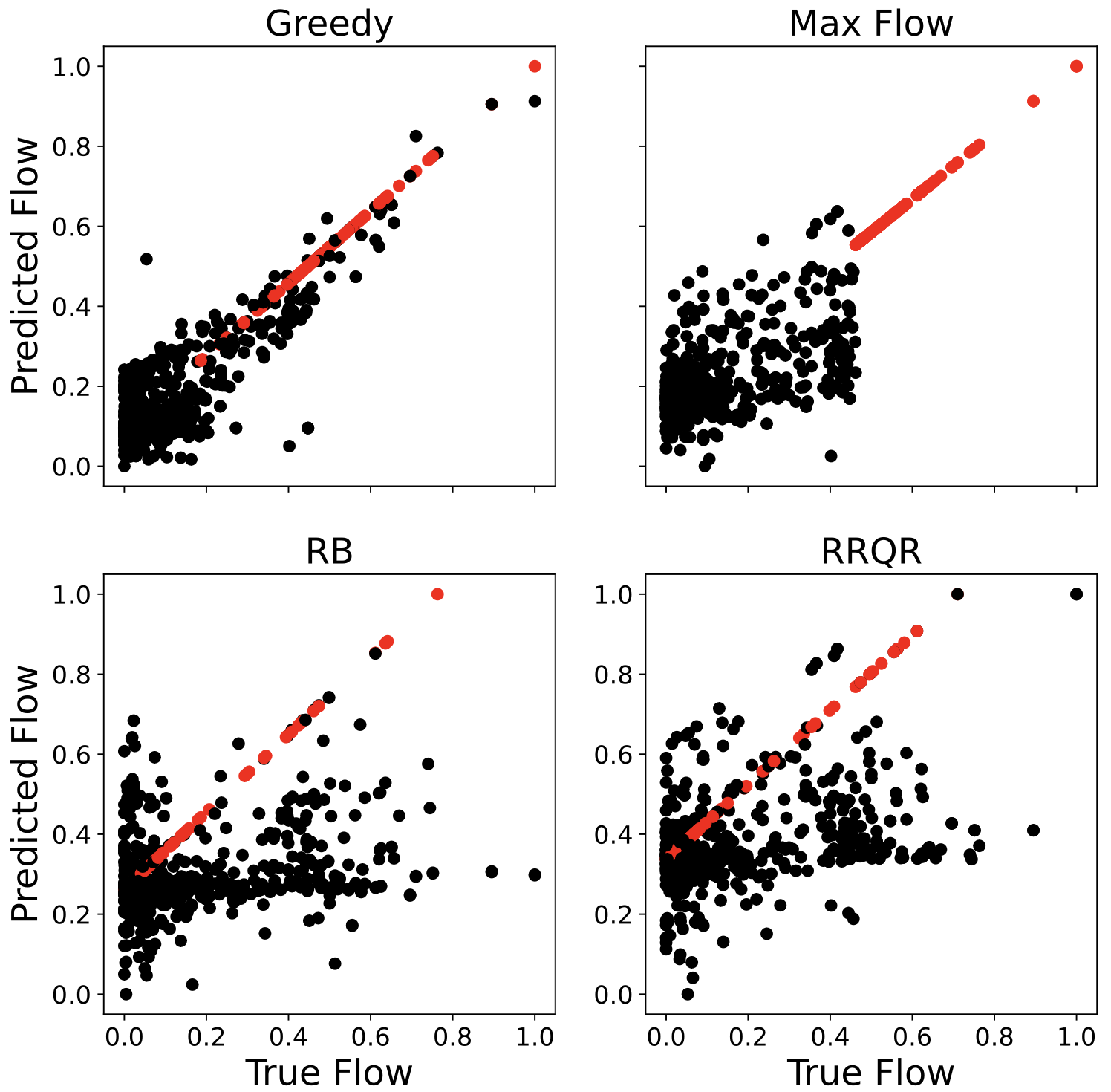}
    \end{center}
    \caption{Predicted vs. ground-truth flows for the Anaheim road network using the flow conservation algorithm (Equation \ref{eqn:flowleastsquares}) for various sensor placement algorithms. Values for the selected sensors are shown in red and inferred values in black. The budget $k$ is fixed at 10\% of all edges. Both the predicted and actual flows are normalized to [0,1]. Our method (Greedy) chooses a more representative set of edges than the baselines.}
    \label{fig::correlationplots}
\end{figure}

\subsubsection{Baselines} We compare our flow selection method with four baselines. First, random selection (Random) simply chooses the next edge uniformly at random from the currently unchosen edges. Second, recursive bisection (RB) \cite{activelearning} partitions the graph using spectral clustering and chooses edges crossing the new cut as the next sensors. The idea is to find "bottleneck" edges where major flows are concentrated, such as highways. Recursive bisection has been found to work well when much of the flow is not conserved \cite{activelearning}. Third, rank-revealing QR (RRQR) \cite{RRQR} exploits a known bound on the error of the flow conservation algorithm. If the SVD of the incidence matrix $\mathbf{B} \in \mathbb{R}^{n \times m}$ is given by $\mathbf{B} = \mathbf{U} \mathbf{\Sigma} \mathbf{V}^T$, then the $m - n + 1$ rightmost columns $\mathbf{V}_C$ of $\mathbf{V}$ give a basis for the cycle space of fully conserved flows. The $n - 1$ leftmost columns $\mathbf{V}_R$ give a basis for the cut space of flows with nonzero divergence. If $S$ is a set of $m - n + 1$ linearly independent rows of $\mathbf{V}_C$ corresponding to the selected sensors, then the reconstruction error is bounded by
\begin{equation}
\label{eqn:errorbound}
    ||\mathbf{\hat{f}} - \mathbf{f}|| \leq (\sigma^{-1}_{\text{min}}(\mathbf{V}_{SC})+1) \cdot ||\mathbf{V}_R\mathbf{f}||
\end{equation}
RRQR uses a greedy heuristic to minimize this bound by minimizing $\sigma^{-1}_{\text{min}}(\mathbf{V}_{SC})$ \cite{activelearning}. 

Finally, maximum flow (Max Flow) selects the $k$ edges with the highest flows:
\begin{equation*}
    S^* = \argmax_{S \subseteq E} ||\mathbf{f}_S||
\end{equation*}
where $\mathbf{f}_S$ is the vector of flows for edges in $S$.

\subsubsection{Evaluation Metrics:} .

\noindent\emph{Correlation (Corr):} 
\begin{equation*}
    Corr=\frac{\textit{cov}(\mathbf{\hat{f}}, f)}{\sigma(\mathbf{\hat{f}}) \cdot \sigma(\mathbf{f})}
\end{equation*}
where \emph{cov} and $\sigma$ are the covariance and std, respectively.

\noindent \emph{Mean squared error (MSE)}:
\begin{equation*}
    MSE=\frac{1}{m} ||\mathbf{f} - \mathbf{\hat{f}}||^2
\end{equation*}
\emph{Mean absolute error (MAE)}:
\begin{equation*}
    MAE=\frac{1}{m} \sum_{e \in E} |\mathbf{f}_e - \mathbf{\hat{f}}_e|
\end{equation*}
\emph{Mean absolute percentage error (MAPE)}:
\begin{equation*}
    MAPE=\frac{1}{m} \sum_{e \in E} |\frac{\mathbf{f}_e - \mathbf{\hat{f}}_e}{\mathbf{f}_e}|
\end{equation*}
\emph{Max error (MAX)}:
\begin{equation*}
    MAX=\max_{e \in E} |\mathbf{f}_e - \mathbf{\hat{f}}_e|
\end{equation*}

\begin{table*}[ht!]
{
\caption{Flow prediction error in terms of Correlation Coefficient (Corr), Mean Squared Error (MSE), Mean Absolute Error (MAE), Mean Absolute Percentage Error (MAPE), and Maximum Error (MAX). Our approach (Greedy) consistently outperforms the baselines on all metrics except for MAPE. }\label{tab:metrics}
\begin{center}
\begin{tabular}{|c|c|c|c|c|c|c|} \hline
%\toprule
%\hline
    & & Corr ($\uparrow$) & MSE ($\downarrow$)& MAE ($\downarrow$)& MAPE($\downarrow$) & MAX ($\downarrow$) \\
    \hline
    \multirow{5}{*}{Anaheim} & Greedy & \textbf{0.936} & \textbf{0.006} & \textbf{0.051} & \textbf{194.381} & \textbf{0.452} \\
    & Max Flow & 0.852 & 0.014  & 0.076 & 264.498 & 0.577 \\
    & RB & 0.402 & 0.043 & 0.130 & 209.277 & 0.963 \\
    & RRQR & 0.493 & 0.038 & 0.121 & 246.388 & 0.818 \\
    & Random & 0.491 & 0.039 & 0.128 & 241.370 & 0.998 \\
    \hline
    \multirow{5}{*}{Barcelona} & Greedy & \textbf{0.869} & \textbf{0.008} & \textbf{0.062} & 289.385 & 0.450 \\
    & Max Flow & 0.847 & 0.009 & 0.066 & 283.837 & \textbf{0.393} \\
    & RB & 0.458 & 0.024 & 0.096 & 275.972 & 0.884 \\
    & RRQR & 0.264 & 0.024 & 0.101 & 181.948 & 0.901 \\
    & Random & 0.303 & 0.029 & 0.099 & \textbf{157.793} & 1.010 \\ 
    \hline
    \multirow{5}{*}{Chicago} & Greedy & \textbf{0.873} & \textbf{0.007} & \textbf{0.059} & 102.544 & 0.323 \\
    & Max Flow & 0.840 & 0.009 & 0.067 & 111.540 & \textbf{0.319} \\
    & RB & 0.220 & 0.023 & 0.096 & 103.362 & 0.992 \\
    & RRQR & 0.184 & 0.024 & 0.099 & 95.762 & 1.000 \\
    & Random & 0.259 & 0.022 & 0.094 & \textbf{92.863} & 0.994 \\
    \hline
    \multirow{5}{*}{Winnipeg} & Greedy & \textbf{0.815} & \textbf{0.013} & \textbf{0.081} & 217.534 & \textbf{0.561} \\
    & Max Flow & 0.785 & 0.015 & 0.087 & 213.350 & 0.591 \\
    & RB & 0.467 & 0.031 & 0.110 & 133.450 & 0.931 \\
    & RRQR & 0.235 & 0.036 & 0.118 & 116.715 &  0.961 \\
    & Random & 0.270 & 0.035 & 0.113 & \textbf{113.394} & 0.986 \\
    %\midrule
    \hline
\end{tabular}
\end{center}
}
\end{table*}

%This can be expanded with one paragraph per plot/table and more discussion
\subsection{Fully-Observed Flows}
We assume access to the ground truth flows $\mathbf{f}$ while choosing the sensors. 

Figure \ref{fig::flowcorrelations} shows the correlation between predicted and ground-truth flows for a varying number of sensors (as a percentage). Our approach significantly outperforms the flow-agnostic baselines, especially when the number of sensors is small (20\% or less). The results also show that there is no clear best baseline among the flow-agnostic methods. In most cases, our algorithm also outperforms Max Flow. The greedy algorithm has the additional advantage of being monotone; that is, adding a sensor greedily almost always improves the prediction. This is not always the case for Max Flow. 

Table \ref{tab:metrics} shows the prediction errors in terms of all the evaluation metrics for a fixed budget of $10\%$ of the candidate sensors being selected. Our greedy solution consistently outperforms the baselines according to most metrics (Corr, MSE, and MAE) and is often competitive with the baselines in terms of MAPE and MAX. We notice that the objective applied by our algorithm is quite different from MAPE (skewed towards small flows) and MAX (skewed towards large flows).

To distinguish our greedy algorithm and Max Flow, Figure \ref{fig::correlationplots} shows the predicted vs ground-truth flows for the Anaheim dataset. The results show the superiority of our approach, especially for intermediate value flows.

%The greedy algorithm outperforms all baselines according to most metrics (see Table \ref{tab:metrics}) and chooses a more intuitive set of labeled edges than any baseline (see Figure \ref{fig::correlationplots}). 

\subsection{Synthetic and Noisy Flows}
\label{sec:unknownflows}

%A possible objection to the previous setting is that, in real-world applications, the complete flows are rarely fully-observed. 

In this section, we remove the assumption that flows are fully observed. 

\subsubsection{Synthetic Flows} We generate synthetic flows under the conservation assumption (see the appendix) \cite{activelearning}. The greedy heuristic is computed using the prediction error on the synthetic flows, and the resulting sensors are tested on the true flows for the four traffic networks. The sensor placements based on synthetic flows do not always outperform the baselines (see Figure \ref{fig::noisyflows}). This is evidence that the synthetic flows are not an effective proxy for the real flows. 

% The skewed distribution of flows explains the difference in performance. Synthetic flows are generated by distributing flow equally among the cycles in $G$ along with some random noise. Real-world flows do not obey this uniform distribution because certain cycles, such as those that include freeways and downtowns, contain much more flow than others. Our greedy heuristic exploits this fact to recover most of the flow by covering only a few cycles. 

\begin{figure}[ht!]
    \centering
     \begin{subfigure}{0.6\textwidth}
        \hspace{0.15\textwidth}
         \includegraphics[width=0.5\textwidth]{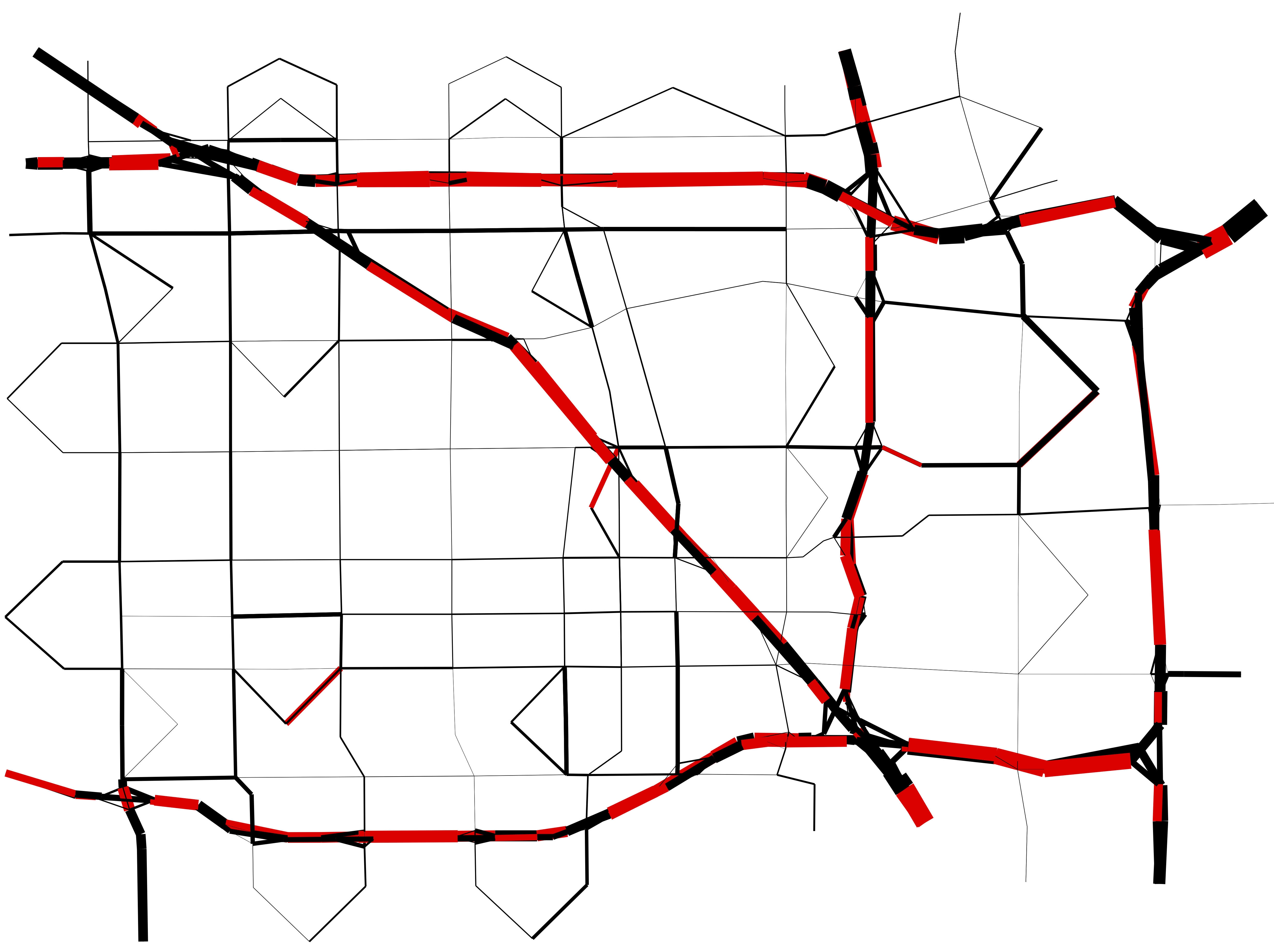}
         \caption{Greedy}
         \label{fig::anaheim_sensors}
     \end{subfigure}
     \hfill
     \begin{subfigure}[b]{0.6\textwidth}
        \hspace{0.15\textwidth}
         \includegraphics[width=0.5\textwidth]{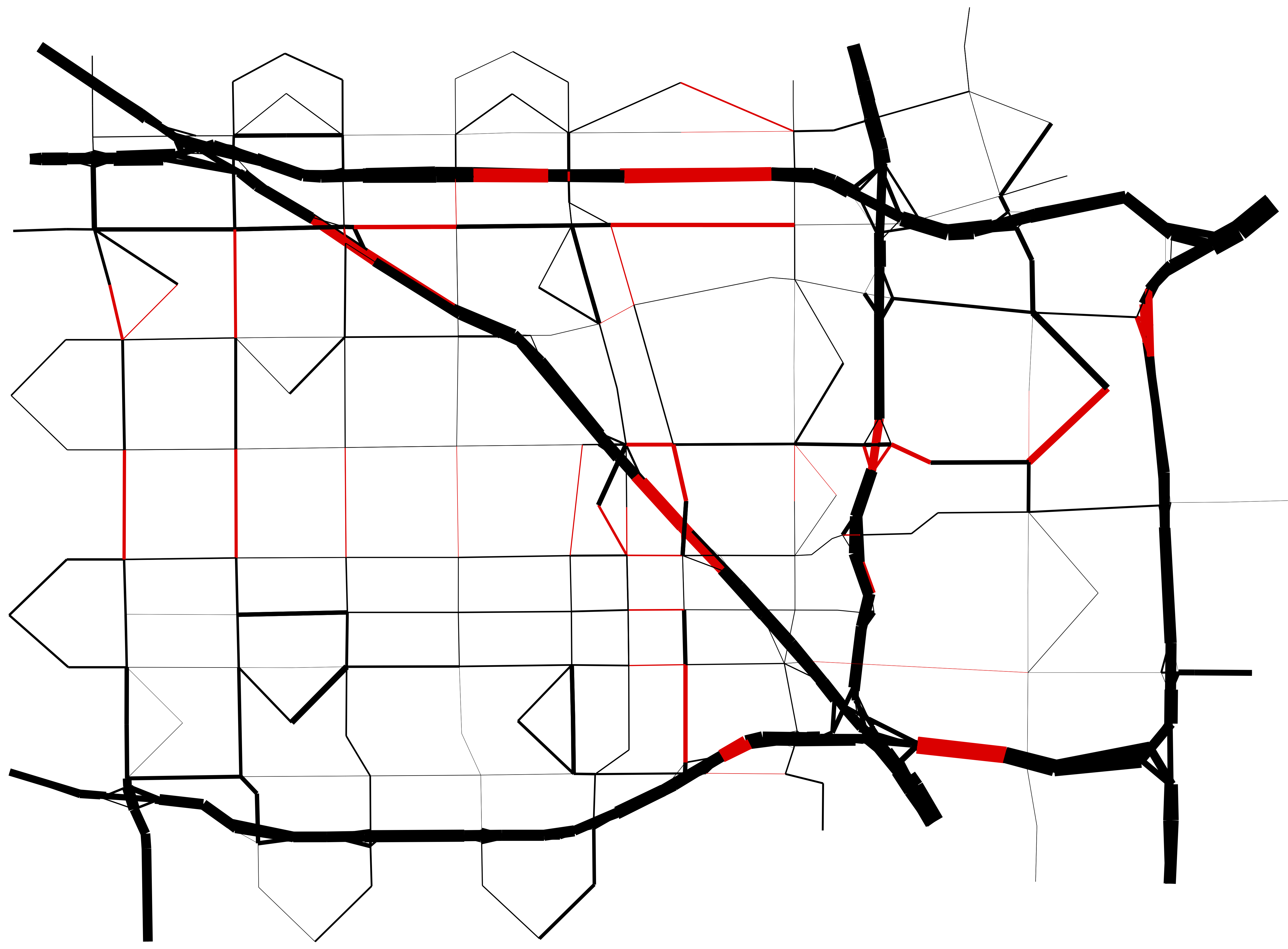}
         \caption{Recursive Bisection}
         \label{fig::anaheim_rb_sensors}
     \end{subfigure}
        \caption{Example of sensor placements (in red) using our heuristic (Greedy) and the Recursive Bisection (RB) baseline for the Anaheim road network. Edge traffic counts are represented by edge thickness. Unlike the baseline, our approach targets a few high-traffic paths.
        }
        \label{fig::sensors}
\end{figure}

\subsubsection{Noisy Estimates} Now we consider the setting where our model has access to noisy estimates of the ground-truth flows (e.g. based on the macro traffic demand model for a city \cite{tamin1989transport,willumsen1981simplified}). For each edge $e_i$, we simulate a noisy estimate $\mathbf{f}_i + \epsilon$ with noise $\epsilon \sim N(0, r\sigma)$ where $\sigma$ is the standard deviation of $\mathbf{f}$ and $r \in \mathbb{R}$ controls the amount of noise. Figure \ref{fig::noisyflows} shows the correlation between the predicted and original ground-truth flows for varying noise levels. The results show that our approach (Noisy Greedy) is robust to noise and often outperforms the baselines. As expected, Max Flow is also robust to noise---as $\epsilon$ is not relative to edge flow values, the noise has little effect on large flows.
%Even with large amounts of noise, our approach (Noisy Greedy) outperforms all baselines. 

\begin{figure}[ht!]
    \centering
    \includegraphics[width=0.45\textwidth]
    {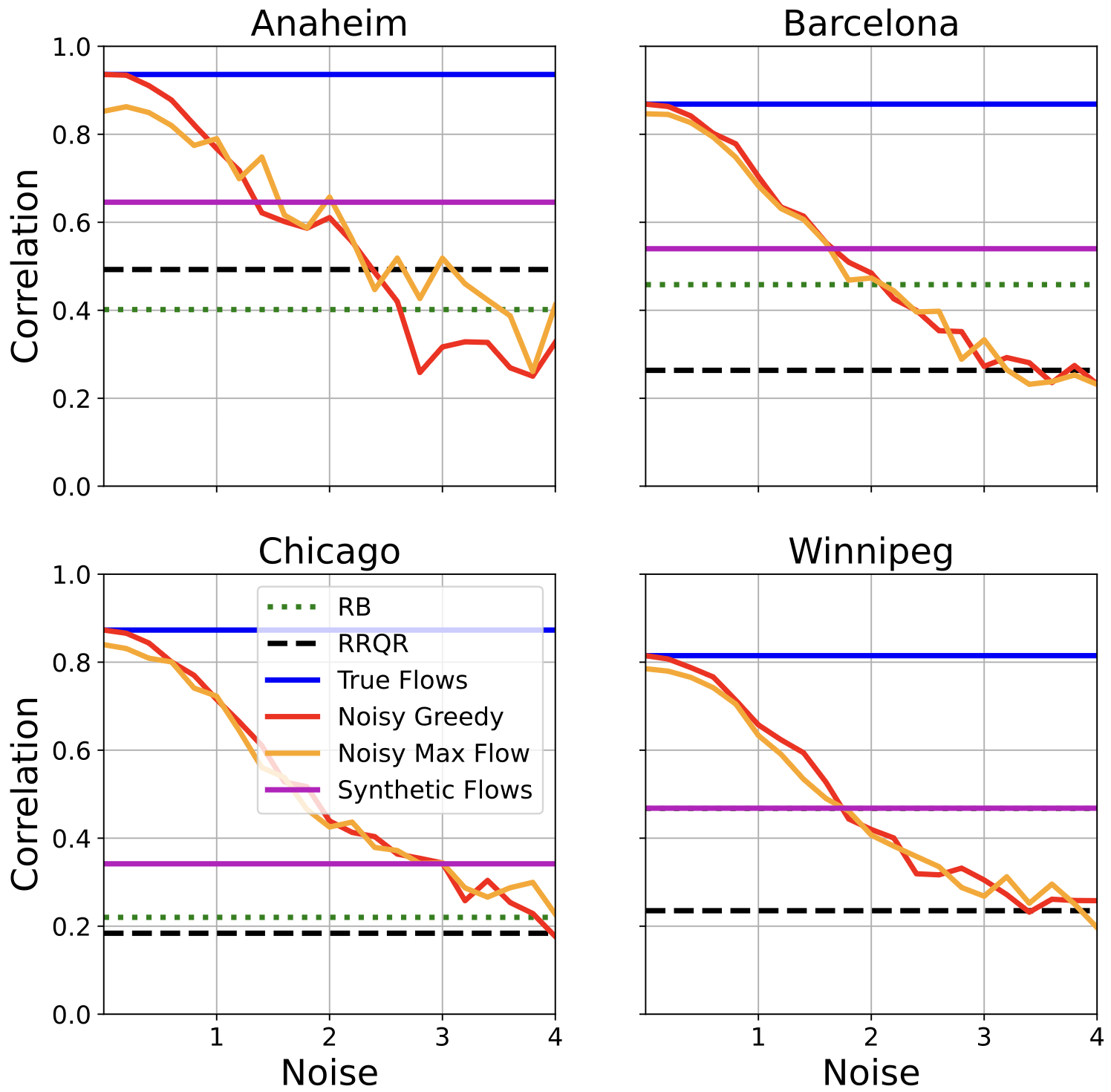}
    \caption{Conserved flow prediction results for sensor selection based on a noisy estimate of ground-truth flows for varying noise levels. The number of sensors placed ($k$) is fixed at 10\% of the edges. The sensor placement is quite robust to noise target values, outperforming the baselines under noise levels of up to 2$\times$ the standard deviation of the flows.
    }
    \label{fig::noisyflows}
\end{figure}

\subsubsection{Visualization} Figure \ref{fig::sensors} shows the sensors placed by our heuristic and the Recursive Bisection (RB) baseline (flow agnostic) for the Anaheim network. We notice that our approach identifies high-traffic paths while also minimizing redundancy.

\subsection{Speedups}

To evaluate the optimizations proposed in Sections \ref{sec:lazyevaluation} and \ref{sec:recursivecomputation}, we run Algorithm \ref{alg::greedy} using a brute force evaluation of the heuristic, with lazy evaluation, and with a combination of lazy evaluation and recursive computation. The brute force evaluation takes over an hour to choose 10\% of the edges in most of the networks tested, but lazy evaluation with recursive computation reduces the running time to a few seconds. Figure \ref{fig::speedups} shows the speedups. 

\begin{figure}[t!]
    \begin{center}
    \includegraphics[width=0.48\textwidth]{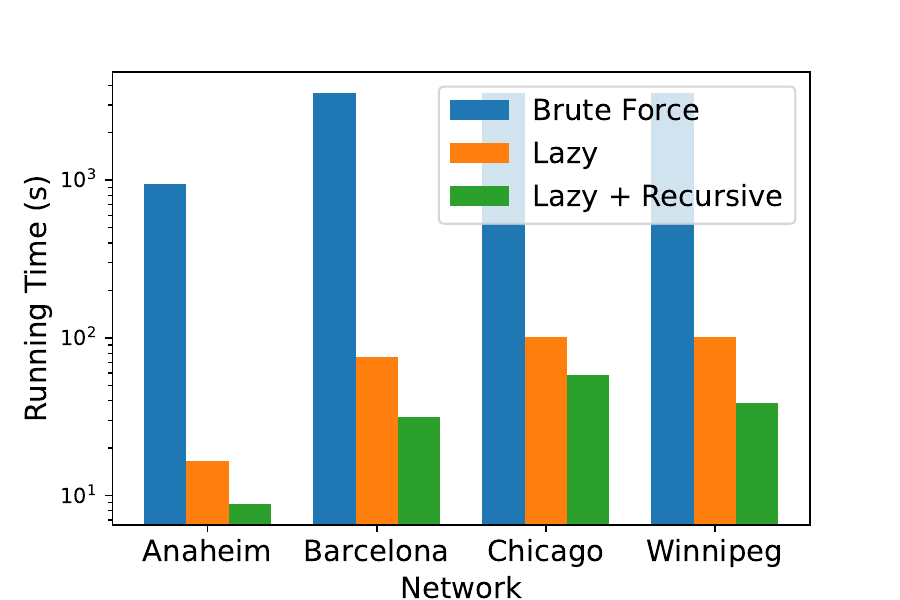}
    \end{center}

    \caption{Running time of Algorithm \ref{alg::greedy} using brute force, lazy evaluation without and with recursive computation for four networks. The number of sensors $k$ is set to 10\% of all edges. In all cases, combining lazy evaluation with the recursive evaluation provided significant speedups. The brute force solution was truncated at 1 hour for all networks except Anaheim.}
    \label{fig::speedups}
\end{figure}

\section{Conclusion and Ongoing Work}

We have proposed an approach for sensor placement for semi-supervised learning in flow networks that accounts for ground-truth values. We show that choosing the optimal set of sensors is NP-hard and provides an effective greedy heuristic for the problem. We also provide two optimizations to further speed-up the proposed heuristic. Our experiments show that our methods significantly outperform baselines that account only for the graph topology and that the proposed optimizations lead to significant running time savings. 

%In ongoing work, we aim to establish theoretical results for the performance of our algorithm. 

Our work opens several venues for future investigation. We are particularly interested in even faster flow-aware sensor placement algorithms based on matrix factorization. This requires the discovery of effective bases for the sparse representation of flows. In terms of applications, we will study how the proposed solutions can be applied to wastewater monitoring. 

\section{Acknowledgements}

This research was partially funded by the 100K Strong IFCE-Rice-SENAI Program on AI for Urban Sustainability and Resilience to Natural Disasters in the Americas, the US Department of Transportation (USDOT) Tier-1 University Transportation Center (UTC) Transportation Cybersecurity Center for Advanced Research and Education (CYBER-CARE), and the Rice Office of Undergraduate Research and Inquiry.

%\clearpage

\bibliographystyle{plain}
\bibliography{bibliography}

% \begin{thebibliography}{99}

% \end{thebibliography}
\end{document}

\newtheorem{definition}{Definition}

\newcommand\relatedversion{}

\title{\Large Appendix \relatedversion}
% \author{Arnav Burudgunte\thanks{Rice University, \{ab141,arlei\}@rice.edu}
% \and Arlei Silva\footnotemark[1]}

\date{}

\maketitle

% Copyright Statement
% When submitting your final paper to a SIAM proceedings, it is requested that you include
% the appropriate copyright in the footer of the paper.  The copyright added should be
% consistent with the copyright selected on the copyright form submitted with the paper.
% Please note that "20XX" should be changed to the year of the meeting.

% Default Copyright Statement
\fancyfoot[R]{\scriptsize{Copyright \textcopyright\ 2024 by SIAM\\
Unauthorized reproduction of this article is prohibited}}

\section{Reproducibility}
\label{appendix:flow}

Here we describe additional setup and implementation details of our experiments. Implementations of all algorithms, along with code for reproducing experiments, can be found at \url{https://github.com/arnav-b/sensor-placement}.

\subsection{Synthetic Flows \cite{activelearning}} We generate synthetic flows based on the singular vector decomposition of $\mathbf{B}$ as follows. In the SVD of $\mathbf{B}$ given above, $\mathbf{U} \in \mathbb{R}^{n \times n}$, $\mathbf{\Sigma} \in \mathbb{R}^{n \times m}$ contains the singular values $\rho_1,...,\rho_m$ of $\mathbf{B}$ along the main diagonal with $m-n$ extra columns of zeros, and $\mathbf{V} \in \mathbb{R}^{m \times m}$ contains the right singular vectors. These vectors $\mathbf{V}_1,...,\mathbf{V}_m$ give a basis for the edge flow space. The synthetic flow $\mathbf{f}$ can be computed as 

\begin{equation}
    \mathbf{f}  = \sum_{i = 1}^m \frac{b}{\rho_i + \epsilon} \mathbf{V}_i
\end{equation}
where $b \in \mathbb{R}$ controls the flow magnitude and $\epsilon \in \mathbb{R}_+$ controls the amount of non-conserved flow. Similar to \cite{activelearning}, we set $b = 20$ and $\epsilon = 0.1$ for our experiments.  

\clearpage

\bibliographystyle{plain}
\bibliography{bibliography}

% \begin{thebibliography}{99}

% \end{thebibliography}